\documentclass[10pt,onecolumn]{IEEEtran}
\usepackage{amsfonts}
\usepackage{amsmath,amssymb}
\usepackage{graphicx}
\usepackage{caption2}
\usepackage{epsfig}
\usepackage[dvips]{color}

\def\dref#1{(\ref{#1})}
\def\rm{\mathrm}

\newtheorem{theorem}{Theorem}
\newtheorem{lemma}{Lemma}
\newtheorem{corollary}{Corollary}
\newtheorem{remark}{Remark}

\begin{document}

\title{Distributed Adaptive Consensus Protocols for Linear Multi-agent Systems
with Directed Graphs and External Disturbances}
%
%
\author{Zhongkui~Li and Zhisheng Duan
\thanks{This work was supported by the National Natural Science Foundation
of China under grants 61104153, 11332001, 61225013, and
a Foundation for the Author of National Excellent Doctoral Dissertation of PR China.}
\thanks{Z. Li and Z. Duan are with the State Key Laboratory for Turbulence
and Complex Systems, Department of Mechanics and Engineering Science, College
of Engineering, Peking University, Beijing 100871, China
(E-mail: zhongkli@pku.edu.cn; duanzs@pku.edu.cn).}}

\maketitle

\begin{abstract}
This paper addresses the distributed consensus design problem for linear multi-agent systems with directed communication graphs
and external disturbances.
Both the cases with strongly connected communication graphs and leader-follower graphs containing a directed spanning
tree with the leader as the root are discussed.
Distributed adaptive consensus protocols based on the relative states of neighboring
agents are
designed, which can ensure
the ultimate boundedness of the consensus error and adaptive gains in the presence of external disturbances.
The upper bounds of the consensus error are further explicitly given.
Compared to the existing consensus protocols,
the merit of the adaptive protocols proposed in this paper is that
they can be computed and implemented in a fully distributed fashion
and meanwhile are robust with respect to external disturbances.
\end{abstract}


\section{Introduction}

In the past two decades, rapid advances in miniaturizing of computing, communication, sensing, and actuation have made it feasible to deploy a large number of autonomous agents to work cooperatively to accomplish civilian and military missions. This, compared to a single complex agent, has the capability to significantly improve the operational effectiveness, reduce the costs, and provide additional degrees of redundancy \cite{ren2007information,olfati-saber2007consensus}.
Having multiple autonomous agents to work together to achieve collective behaviors is usually referred to as cooperative control of multi-agent systems. Due to its potential applications in various areas such
as
surveillance and reconnaissance systems, satellite formation flying, electric power systems,
and intelligent transportation systems, cooperative control of multi-agent systems has
received compelling attention from the systems and control
community. In the area of cooperative control, consensus is an important and fundamental problem,
which means that a group of agents reaches an agreement on certain quantity of
interest by interacting with their local neighbors.
Advances of various consensus
algorithms have been reported in a quite large body of
research papers; see \cite{ren2007information,hong2008distributed,olfati-saber2004consensus,li2010distributed,Guo2013consensus,li2013distributed,carli2009quantized,RNC:RNC2904,RNC:RNC2867}
and the references therein.

For a consensus control problem, the main task is to design appropriate protocols to achieve consensus.
Due to the large number of agents,
limited sensing capability of sensors, and short wireless communication ranges,
distributed control, depending only on local information of the agents
and their neighbors, appears to be a promising tool for handling multi-agent
systems. Note that designing appropriate distributed protocols is generally a challenging task,
especially for multi-agent systems with complex dynamics,
due to the interplay of the agent dynamics, the communication graph among agents,
and the cooperative control laws.

Take the consensus problem for multi-agent systems
with general continuous-time linear dynamics for instance.
In previous works \cite{li2010consensus,li2011dynamic,tuna2009conditions,seo2009consensus,zhang2011optimal,ma2010necessary},
several static and dynamic consensus protocols based on the local state or output information of neighboring agents
have been proposed.
A common feature in the aforementioned existing papers
is that the design of the consensus protocols needs to use
some eigenvalue information of the Laplacian matrix associated with the
communication graph.
Actually, for the simply case with second-order integrator agent dynamics,
the design of the consensus
protocols does rely on the smallest real part
of the nonzero eigenvalues of the Laplacian matrix \cite{tuna2009conditions}.
However, the smallest real part of the nonzero eigenvalues of the Laplacian matrix is global information of the communication graph,
because every agent has to know the whole communication graph to compute it.
Therefore, the consensus protocols in the aforementioned works
cannot be designed by the agents in a fully distributed way, i.e.,
relying on only the local information of neighboring agents.
To avoid this limitation,
two types of distributed adaptive consensus protocols
are proposed in \cite{li2012adaptiveauto,li2011adaptive}, which implement adaptive laws to
dynamically update the coupling weights of neighboring agents.
Similar adaptive schemes are presented in \cite{su2011adaptive,yu2013distributed} to achieve
consensus for multi-agent systems with second-order nonlinear dynamics. Note that the adaptive protocols in
\cite{li2012adaptiveauto,li2011adaptive,su2011adaptive,yu2013distributed}
are applicable to only undirected communication graphs or leader-follower
graphs where the subgraph
among followers is undirected.
Because of the asymmetry of the Laplacian matrices of directed graphs,
designing fully distributed adaptive consensus protocols for
general directed graphs is more much challenging.
By introducing monotonically increasing functions to provide
additional freedom for design,
a distributed adaptive consensus protocol is constructed in \cite{li2014TAC}
to achieve leader-follower consensus for
the communication graphs containing a directed spanning
tree with the leader as the root node. Similar adaptive protocols for directed graphs are developed
in \cite{mei2014consensus} for the special case where the agents are described by double integrators.
Even though the aforementioned advances have been reported on designing distributed adaptive protocols for general linear multi-agent systems
with directed graphs, there are still many important open problems awaiting further investigation. For instance,
to design distributed adaptive protocols for the case of general directed graphs without a leader,
to examine the robustness issue associated with the adaptive protocols, and to
propose distributed dynamic adaptive protocols for the case where only local output information is available,
to name just a few.

In this paper, we intend to address the first two aforementioned problems. Specifically, we address the distributed adaptive consensus
protocol design problem for general linear multi-agent systems
with directed graphs and external disturbances. We consider both the cases where
the communication graph among the agents is strongly connected
and contains a directed spanning tree with the leader as the root.
For the case with leader-follower graphs containing
a directed spanning tree with the leader as the root,
we revisit the distributed adaptive protocol in \cite{li2014TAC}.
It is pointed out that in the presence of external disturbances,
the adaptive gains of the adaptive protocol in \cite{li2014TAC} will slowly grow to infinity,
which is the well-known parameter drift phenomenon in the adaptive
control literature \cite{ioannou1996robust}.
To deal with this instability issue associated with the adaptive protocol in \cite{li2014TAC},
we propose a novel distributed adaptive consensus protocol,
by using $\sigma$ modification technique \cite{ioannou1996robust}.
This novel adaptive protocol
is designed in fully distributed fashion
to ensure the ultimate boundedness of both the consensus error
and the adaptive coupling gains. That is, the proposed adaptive
protocol is robust in the presence of external disturbances.
The upper bound of the consensus error is also explicitly given.
The case with strongly connected communication graphs is further studied.
A distributed robust adaptive protocol is also presented, which can
guarantee the ultimate boundedness of the consensus error
and the adaptive coupling gains in the presence of external disturbances.
A sufficient condition for the existence of
the adaptive protocols proposed in this paper is that each agent
is stabilizable.

The rest of this paper is organized as follows.
Mathematical preliminaries required in this paper is summarized in Section 2.
Distributed robust adaptive consensus protocols are presented in Sections 3 and 4 for multi-agent systems
with strongly connected graphs and directed leader-follower graphs, respectively.
Simulation examples are presented
for illustration in Section 5. Conclusions are
drawn in Section 6.

\section{Mathematical Preliminaries}
In this paper, we use the following notations and definitions:
$\mathbf{R}^{n\times m}$ represents the set of  $n\times m$
real matrices.
$I_N$ denotes the identity matrix of
dimension $N$.
$\mathbf{1}$ denotes a column vector of appropriate dimension
with its entries equal to one. 
For real symmetric matrices $W$
and $X$, $W>(\geq)X$ means that $W-X$ is positive (semi-)definite.
$A\otimes B$ denotes the Kronecker product of the matrices $A$ and $B$.
Denote by $\sigma_{\max}(B)$ the largest singular value
of a matrix $B$.
A matrix $A =[a_{ij} ]\in \mathbf{R}^n$ is
termed as a nonsingular $M$-matrix, if $a_{ij} < 0$, $ \forall i \neq j$,
and all the eigenvalues of $A$ have positive real parts.

A directed communication graph $\mathcal {G}$ is a pair $(\mathcal {V}, \mathcal
{E})$, where $\mathcal {V}=\{v_1,\cdots,v_N\}$ is a nonempty finite
set of vertices (i.e., nodes) and $\mathcal {E}\subseteq\mathcal {V}\times\mathcal
{V}$ is a set of edges, where an edge is represented by an
ordered pair of distinct vertices. For an edge $(v_i,v_j)$, $v_i$
is called the parent vertex, $v_j$ is called the child vertex, and $v_i$ is a
neighbor of $v_j$. If a directed graph having the property that
$(v_i,v_j)\in\mathcal {E}$ implies $(v_j, v_i)\in\mathcal {E}$ for
any $v_i,v_j\in\mathcal {V}$, then it is an undirected graph.
A directed path from
vertex $v_{i_1}$ to vertex $v_{i_l}$ is a sequence of ordered edges in
the form $(v_{i_k}, v_{i_{k+1}})$, $k=1,\cdots,l-1$. 
A directed graph contains a directed spanning tree if
there exists a vertex called the root, which has no parent vertex, such
that there exist directed paths from the vertex to all other vertices in the graph.
A directed graph is strongly connected if there is a directed
path between every distinct vertices.
A directed graph has a directed spanning tree if it is strongly connected,
but not vice versa.

For the directed graph $\mathcal {G}$, its
adjacency matrix $\mathcal {A}=[a_{ij}]\in\mathbf{R}^{N\times
N}$ is defined such that
$a_{ii}=0$, $a_{ij}=1$ if $(v_j,v_i)\in\mathcal {E}$ and $a_{ij}=0$
otherwise. The Laplacian matrix $\mathcal {L}=[\mathcal
{L}_{ij}]\in\mathbf{R}^{N\times N}$ associated with $\mathcal {G}$
is defined as $\mathcal
{L}_{ii}=\sum_{j\neq i}a_{ij}$ and $\mathcal {L}_{ij}=-a_{ij}$,
$i\neq j$. 

\begin{lemma} [\cite{ren2005consensus}]\label{lemma1}
Zero is an eigenvalue of $\mathcal {L}$ with $\mathbf{1}$ as a
right eigenvector and all the nonzero eigenvalues have positive real
parts. Zero is a simple eigenvalue of $\mathcal {L}$ if
and only if $\mathcal {G}$ contains a directed spanning tree.
\end{lemma}

\begin{lemma}[\cite{mei2014consensus}]\label{lemma2}
Suppose that $\mathcal {G}$ is strongly connected.
Let $r=[r_1,\cdots,r_N]$ be the positive left eigenvector of $\mathcal
{L}$ associated with the zero eigenvalue
and $R={\rm{diag}}(r_1,\cdots,r_N)$. Then, $\widehat{\mathcal {L}}\triangleq R\mathcal {L}+\mathcal {L}^TR$
is the symmetric Laplacian matrix associated with an
undirected connected graph.
Let $\xi$ by any vector with positive entries.
Then,
$\min_{\xi^Tx=0,x\neq0}\frac{x^T\widehat{\mathcal {L}}x}{x^Tx}>\frac{\lambda_2(\widehat{\mathcal {L}})}{N},$
where $\lambda_2(\widehat{\mathcal {L}})$ denotes the smallest nonzero eigenvalue of $\widehat{\mathcal {L}}$.
\end{lemma}

\begin{lemma}[\cite{li2014TAC}]\label{lemma3}
For a nonsingular $M$-matrix $M$, there exists a positive diagonal matrix $G$ such that $GM+M^TG>0$.
Moreover, $G$ can be given by ${\mathrm{diag}}(q_1,\cdots,q_{N})$, where $q =[q_1,\cdots,q_{N}]^T=(M^T)^{-1}{\bf 1}$.
\end{lemma}

\begin{lemma}[Young's Inequality, \cite{bernstein2009matrix}]\label{lemma4}
If $a$ and $b$ are nonnegative real numbers and $p$ and $q$ are positive real numbers such that $1/p + 1/q = 1$, then
$ab\leq\frac{a^p}{p}+\frac{b^q}{q}$.
\end{lemma}

\begin{lemma}[\cite{corless1981continuous}]\label{lemma5}
For a dynamical system $\dot{x}=f(x,t)$,
where $f(\cdot,\cdot)$ is
locally Lipschitz in $x$ and piecewise continuous in $t$,
suppose that there exists
a continuous and differentiable function $V(x,t)$ such that
along any trajectory of the system,
$$\begin{aligned}
\alpha_1(\|x\|)\leq V(x,t)\leq \alpha_2(\|x\|),\\
\dot{V}(x,t)\leq-\alpha_3(\|x\|)+\epsilon,
\end{aligned}$$
where $\epsilon$ is a positive constant, $\alpha_1$ and $\alpha_2$ are class $\mathcal {K}_\infty$ functions,
and $\alpha_3$ is a class $\mathcal {K}$ function. 
Then, the solution
$x(t)$ of $\dot{x}=f(x,t)$ is uniformly ultimately bounded.
\end{lemma}

\section{Distributed Robust Adaptive Protocols for Leader-follower Graphs in The Presence of External Disturbances}

\subsection{Problem Formulation and Motivation}

In this paper, we consider a group of $N$ identical agents with general
linear dynamics.
The dynamics of the $i$-th agent are described by
\begin{equation}\label{1c}
\begin{aligned}
    \dot{x}_i &=Ax_i+B[u_i+\omega_i],
\quad i=1,\cdots,N,
\end{aligned}
\end{equation}
where $x_i\in\mathbf{R}^n$ is the state vector,
$u_i\in\mathbf{R}^{p}$ is the control input vector,
$A$ and $B$ are known constant matrices with
compatible dimensions, and
$\omega_i\in\mathbf{R}^n$
denotes external disturbances associated with the $i$-th agent,
which satisfies the matching condition and the following assumption.

{\it Assumption 1:} There exist
positive constants $\upsilon_i$ such that
$\|\omega_i\|\leq\upsilon_i$, $i=1,\cdots,N$.

In this section, we consider the case where there are
$N-1$ followers and one leader.
Without loss of generality, let the agent in \dref{1c} indexed by 1
be the leader whose control input is assumed to be zero, i.e., $u_1=0$,
and the agents indexed by $2,\cdots,N$, be the
followers.
The communication graph $\mathcal {G}$ among the $N$ agents is assumed to satisfy
the following assumption.

{\it Assumption 2:}
The graph $\mathcal {G}$ contains a directed spanning tree with the leader as the root node.

Since the vertex indexed by 1
is the leader and it does not have any neighbor, the Laplacian matrix $\mathcal {L}$ associated with
$\mathcal {G}$ can be written into
\begin{equation}\label{lapc}
\mathcal {L}=\begin{bmatrix} 0 & 0_{1\times (N-1)} \\
\mathcal {L}_2 & \mathcal {L}_1\end{bmatrix},
\end{equation}
where $\mathcal {L}_2\in\mathbf{R}^{(N-1)\times 1}$ and $\mathcal
{L}_1\in\mathbf{R}^{(N-1)\times (N-1)}$. In light of Lemma \ref{lemma1},
it is easy to see that $\mathcal {L}_1$ is a nonsingular $M$-matrix
if $\mathcal {G}$ satisfies Assumption 2.

This intention of this section is to solve the consensus problem for the agent in \dref{1c},
i.e., to design
distributed consensus protocols under which the states of the $N-1$ followers
converge to the state of the leader in the sense of
$\lim_{t\rightarrow \infty}\|x_i(t)- x_1(t)\|=0$,
$\forall\,i=2,\cdots,N.$

For the case without external disturbances, i.e.,
$\omega_i=0$, $i=1,\cdots,N$,
a distributed adaptive consensus protocol
based on the relative states of neighboring agents was proposed
for each follower as \cite{li2014TAC}:
\begin{equation}\label{ssda0}
\begin{aligned}
u_i &=\bar{c}_i\rho_i(\xi_i^TQ\xi_i)K\xi_i,\\
\dot{\bar{c}}_i &=\xi_i^T\Gamma\xi_i,\quad i=2,\cdots,N,
\end{aligned}
\end{equation}
where $\xi_i\triangleq\sum_{j=1}^Na_{ij}(x_i-x_j)$, $, i=2,\cdots,N$,
$\bar{c}_i(t)$ denotes the time-varying coupling gain (weight) associated with the $i$-th follower
with $\bar{c}_i(0)\geq1$,
$a_{ij}$ is the $(i,j)$-th entry of the adjacency matrix $\mathcal {A}$ associated with
$\mathcal {G}$,
$K\in\mathbf{R}^{p\times n}$ and $\Gamma\in\mathbf{R}^{n\times n}$
are the feedback gain matrices,
$\rho_i(\cdot)$ are smooth
and monotonically increasing functions
which satisfies that
$\rho_i(s)\geq1$ for $s>0$,
and $Q>0$ is a
solution to the following algebraic Riccati equation (ARE):
\begin{equation}\label{alg1}
A^TQ+QA+I-QBB^TQ=0.
\end{equation}

\begin{lemma}[\cite{li2014TAC}]\label{lems2}
Suppose that the communication graph $\mathcal
{G}$ satisfies Assumption 2.
Then, the consensus problem of the agents
in \dref{1c} is solved by the adaptive protocol \dref{ssda0}
with $K=-B^TQ$, $\Gamma=QBB^TQ$,
and $\rho_i(\xi_i^TQ\xi_i)=(1+\xi_i^TQ\xi_i)^3$.
Moreover, each coupling gain
$\bar{c}_{i}$ converges to some finite steady-state value.
\end{lemma}

Lemma \ref{lems2} shows that
the adaptive protocol \dref{ssda0} can achieve consensus for the
case where the agents in \dref{1c} are not perturbed by external disturbances.
The adaptive protocol \dref{ssda0} constructed by Lemma \ref{lems2},
contrary to the consensus protocols in \cite{li2010consensus,li2011dynamic,seo2009consensus,tuna2008lqr,zhang2011optimal},
depends on only the agent dynamics and the relative states of neighboring
agents, and thereby can be computed and implemented by each agent
in a fully distributed way.

Note that in many circumstances, the agents may be subject to certain external disturbances,
for which case it is necessary and interesting to investigate whether the adaptive protocol
\dref{ssda0} designed by Lemma \ref{lems2} is still applicable, i.e.,
whether \dref{ssda0} is robust with respect
to external disturbances. Due to the existence of
nonzero $\omega_i$ in \dref{1c},
the relative states will not converge to zero any more
but rather can only be expected to converge
into some small neighborhood of the origin.
Since the derivatives of the adaptive gains $\bar{c}_i$
are of nonnegative quadratic forms
in terms of the relative states, in this case it is easy to see from \dref{ssda0} that
$\bar{c}_{i}$ will keep growing to infinity,
which is known as the parameter drift phenomenon in the classic adaptive
control literature \cite{ioannou1996robust}.
Therefore, the adaptive protocol
\dref{ssda0} is not robust
in the presence of external disturbances.

\subsection{Distributed Robust Adaptive Consensus Protocols}

The objective of this subsection is to made
modifications on \dref{ssda0}
in order to present novel distributed
robust adaptive protocols which can
guarantee the stability or ultimate boundedness
of the consensus error in the presence of
external disturbances.
By utilizing the $\sigma$ modification technique \cite{ioannou1996robust},
we propose a new distributed adaptive consensus protocol as follows:
\begin{equation}\label{ssda}
\begin{aligned}
u_i &=c_i\rho_i(\xi_i^TQ\xi_i)K\xi_i,\\
\dot{c}_i &=-\varphi_i(c_i-1)+\xi_i^T\Gamma\xi_i,\quad i=2,\cdots,N,
\end{aligned}
\end{equation}
where
$\varphi_i$, $i=2,\cdots,N$, are small positive constants
and the rest of the variables are defined as in \dref{ssda0}.

Let $\xi=[\xi_2^T,\cdots,\xi_N^T]^T$, where $\xi_i=\sum_{j=1}^Na_{ij}(x_i-x_j)$. Then,
it follows that
\begin{equation}\label{conerr}
\xi=(\mathcal {L}_1\otimes I_n)\begin{bmatrix} x_2-x_1\\\vdots\\x_N-x_1\end{bmatrix},
\end{equation}
where $\mathcal {L}_1$
is defined as in \dref{lapc}. Because $\mathcal {L}_1$
is a nonsingular matrix for $\mathcal
{G}$ satisfying Assumption 2, the consensus problem
is solved if and only if $\xi$ asymptotically
converges to zero.
Hereafter, $\xi$ is referred to as the consensus error.
In light of \dref{1c} and \dref{ssda},
it is not difficult to verify that
$\xi$ and $c_i$ satisfy the following dynamics:
\begin{equation}\label{netss1}
\begin{aligned}
\dot{\xi}
&= [I_{N-1}\otimes A+\mathcal {L}_1\widehat{C}\hat{\rho}(\xi)\otimes BK]\xi+(\mathcal {L}_1\otimes B)\omega,\\
\dot{c}_i &=-\varphi_i(c_i-1)+\xi_i^T\Gamma\xi_i,
\end{aligned}
\end{equation}
where $\widetilde{\omega}\triangleq [\omega_2^T-\omega_1^T,\cdots,\omega_N^T-\omega_0^T]^T$,
$\hat{\rho}(\xi)\triangleq{\rm{diag}}(\rho_2(\xi_2^TQ\xi_2),\cdots,\rho_N(\xi_N^TQ\xi_N))$, and
$\widehat{C}\triangleq{\rm{diag}}(c_2,\cdots,c_N)$.

In light of Assumption 1, it is easy to see that
\begin{equation}\label{omega}
\|\widetilde{\omega}\|\leq\sqrt{\sum_{i=2}^N(\upsilon_i+\upsilon_1)^2}.
\end{equation}
By the second equation in \dref{netss1}, we have
\begin{equation}\label{cc0}
\begin{aligned}
c_i(t)
&= c_i(0)\mathrm{e}^{-\varphi_i t}+\int_0^t\mathrm{e}^{-\varphi_i (t-s)}(\varphi_i +\xi_i^T\Gamma\xi_i) ds\\
&= (c_i(0)-1)\mathrm{e}^{-\varphi_i t}+1+\int_0^t\mathrm{e}^{-\varphi_i (t-s)}\xi_i^T\Gamma\xi_ids\\
&\geq 1,
\end{aligned}
\end{equation}
where we have used the fact that $c_i(0)\geq 1$ to get the last inequality.

The following theorem designs the adaptive protocol \dref{ssda}.

\begin{theorem}\label{s2thm1}
Suppose that the communication graph $\mathcal
{G}$ satisfies Assumption 2.
Then, both the consensus error $\xi$
and the coupling gains $c_i$, $i=2,\cdots,N$, in \dref{netss1},
under the adaptive protocol \dref{ssda} with $K$, $\Gamma$, and $\rho_i(\cdot)$ designed as in Lemma \ref{lems2},
are uniformly ultimately bounded.
Furthermore, if $\varphi_i$ is chosen to be small enough such that
$\delta\triangleq\underset{i=1,\cdots,N}{\min}\varphi_i<\tau\triangleq \frac{1}{\lambda_{\max}(Q)}$,
then $\xi$ exponentially converges to the residual set
\begin{equation}\label{d2k}
\mathcal {D}_1\triangleq \left\{\xi:\|\xi\|^2\leq
\frac{2\Pi}{(\tau-\delta)\lambda_{\min}(Q)\underset{i=2,\cdots,N}{\min} q_i}\right\},
\end{equation}
where
\begin{equation}\label{kk}
\Pi\triangleq\frac{\hat{\lambda}_0}{24}\sum_{i=2}^N\varphi_i(\alpha-1)^2+\frac{12}{\hat{{\lambda}}_0}\sigma_{\max}^2(G \mathcal {L}_1)\sum_{i=2}^N(\upsilon_i+\upsilon_1)^2,
\end{equation}
$[q_2,\cdots,q_{N}]^T=(\mathcal {L}_1^T)^{-1}{\bf 1}$,
$G={\rm{diag}}(q_2,\cdots,q_{N})$,
$\alpha=\frac{72}{\hat{\lambda}_0^2}\underset{i=2,\cdots,N}{\max} q_i^2+\underset{i=2,\cdots,N}{\max} \frac{2 q_i^3}{\hat{\lambda}_0^3}$,
and $\hat{\lambda}_0$
denotes the smallest eigenvalue of $G\mathcal {L}_1+\mathcal {L}_1^TG$.
\end{theorem}

\begin{proof}
Consider the following Lyapunov function candidate:
\begin{equation}\label{lyas1}
V_1=\sum_{i=2}^N\frac{c_iq_i}{2}\int_0^{\xi_i^TQ\xi_i}\rho_i(s)ds+\frac{\hat{\lambda}_0}{24}\sum_{i=2}^N\tilde{c}_i^2,
\end{equation}
where $\tilde{c}_i=c_i-\alpha$.
As mentioned earlier, for a communication graph $\mathcal {G}$ satisfying Assumption 2,
$\mathcal {L}_1$ is a nonsingular $M$-matrix, which, by Lemma \ref{lemma3}, implies that
$G>0$ and $\hat{\lambda}_0>0$.
Furthermore, by noting that
$\rho_i(\cdot)$ are monotonically increasing functions
satisfying $\rho_i(s)\geq1$ for $s>0$ and that
$c_i(t)\geq1$ for $t>0$ as in \dref{cc0}, it is not hard to get
that $V_1$ is positive definite with respect to $\xi_i$ and $\tilde{c}_i$, $i=2,\cdots,N$.

The time derivative of $V_1$
along the trajectory of \dref{netss1} can be obtained as
\begin{equation}\label{lyas2}
\begin{aligned}
\dot{V}_1 &=
\sum_{i=2}^N c_i q_i \rho_i(\xi_i^TQ\xi_i)\xi_i^TQ\dot{\xi}_i\\
&\quad+
\sum_{i=2}^N\frac{\dot{c}_iq_i}{2}\int_0^{\xi_i^TQ\xi_i}\rho_i(s)ds\\
&\quad
+\frac{\hat{\lambda}_0}{12}\sum_{i=2}^N(c_i-\alpha)[-\varphi_i(c_i-1)+\xi_i^T\Gamma\xi_i].
\end{aligned}
\end{equation}
In the rest of this proof, we will use $\hat{\rho}$
and $\rho_i$ instead of $\hat{\rho}(\xi)$
and $\rho_i(\xi_i^TQ\xi_i)$, respectively,
whenever without causing any confusion.

Observe that
\begin{equation}\label{lyas3}
\begin{aligned}
\sum_{i=2}^N c_iq_i \rho_i\xi_i^TQ\dot{\xi}_i&=\xi^T(\widehat{C}\hat{\rho} G \otimes Q)\dot{\xi}\\
&=\frac{1}{2}\xi^T[\widehat{C}\hat{\rho} G \otimes (QA+A^TQ)\\
&\quad-
\widehat{C}\hat{\rho} (G\mathcal {L}_1+\mathcal {L}_1^TG)\widehat{C}\hat{\rho} \otimes QBB^TQ]\xi\\
&\quad +\xi^T (\widehat{C}\hat{\rho} G \mathcal {L}_1\otimes QB)\widetilde{\omega}\\
&\leq\frac{1}{2}\xi^T[\widehat{C}\hat{\rho} G \otimes (QA+A^TQ)-\hat{\lambda}_0
\widehat{C}^2\hat{\rho}^2 \otimes QBB^TQ]\xi\\
&\quad+\xi^T (\widehat{C}\hat{\rho} G \mathcal {L}_1\otimes QB)\widetilde{\omega},
\end{aligned}
\end{equation}
where we have used the fact that $G\mathcal {L}_1+\mathcal {L}_1^TG>0$ to get the first inequality.

Since $\rho_i$ are monotonically
increasing functions satisfying that $\rho_i(s)\geq1$ for $s>0$, it follows that
\begin{equation}\label{lyas4}
\begin{aligned}
\sum_{i=2}^N \dot{c}_iq_i \int_0^{\xi_i^TQ\xi_i}\rho_i(s)ds&
\leq \sum_{i=2}^N \dot{\bar{c}}_i q_i \int_0^{\xi_i^TQ\xi_i}\rho_i(s)ds\\
&\leq \sum_{i=2}^N \dot{\bar{c}}_i q_i \rho_i\xi_i^TQ\xi_i\\
&\leq \sum_{i=2}^N\frac{\dot{\bar{c}}_iq_i^3}{3\hat{\lambda}_0^2}+\sum_{i=2}^N\frac{2}{3}\hat{\lambda}_0\dot{\bar{c}}_i\rho_i^{\frac{3}{2}}(\xi_i^TQ\xi_i)^{\frac{3}{2}}\\
&\leq \sum_{i=2}^N\frac{\dot{\bar{c}}_iq_i^3}{3\hat{\lambda}_0^2}+\sum_{i=1}^N\frac{2}{3}\hat{\lambda}_0\dot{\bar{c}}_i\rho_i^{\frac{3}{2}}(1+\xi_i^TQ\xi_i)^{\frac{3}{2}}\\
&\leq \sum_{i=1}^N(\frac{q_i^3}{3\hat{\lambda}_0^2}+\frac{2}{3}\hat{\lambda}_0\rho_i^2)(\xi_i^TQBB^TQ\xi_i),
\end{aligned}
\end{equation}
where we have used the fact that $\dot{c}_i\leq \dot{\bar{c}}_i$ ($\dot{\bar{c}}_i$ is defined in \dref{ssda0})
to get the first inequality,
used the well-known mean value theorem for integrals
to obtain the second inequality, and used Lemma \ref{lemma4} to get the third inequality.

Substituting \dref{lyas3} and \dref{lyas4} into \dref{lyas2} yields
\begin{equation}\label{lyas5}
\begin{aligned}
\dot{V}_1 &\leq
\frac{1}{2}\xi^T[\widehat{C}\hat{\rho} G \otimes (QA+A^TQ)]\xi\\
&\quad-\sum_{i=1}^N
[\hat{\lambda}_0(\frac{1}{2}c_i^2\rho_i^2-\frac{1}{12}c_i-\frac{1}{3}\rho_i^2)+\frac{1}{12}(\hat{\lambda}_0\alpha-
\frac{2 q_i^3}{\hat{\lambda}_0^2})]\xi^T_iQBB^TQ\xi_i\\
&\quad +\frac{\hat{\lambda}_0}{24}\sum_{i=1}^N\varphi_i[-\tilde{c}_i^2+(\alpha-1)^2]+\xi^T (\widehat{C}\hat{\rho} G \mathcal {L}_1\otimes QB)\widetilde{\omega},
\end{aligned}
\end{equation}
where we have used the following fact:
\begin{equation}\label{sskk2}
-(c_i-\alpha)(c_i-1)=-\tilde{c}_i(\tilde{c}_i+\alpha-1)\leq -\frac{1}{2}
\tilde{c}_i^2+\frac{1}{2}(\alpha-1)^2.
\end{equation}
By noting that $\rho_i\geq 1$
and $c_i\geq1$, $i=1,\cdots,N$, and
that $\alpha=2\hat{\alpha}+\underset{i=1,\cdots,N}{\max}\frac{2 q_i^3}{\hat{\lambda}_0^3}$, where $\hat{\alpha}=\frac{ 36}{\hat{\lambda}_0^2}\underset{i=1,\cdots,N}{\max} q_i^2$,
it follows from \dref{lyas5}
that
\begin{equation}\label{lyas7}
\begin{aligned}
\dot{V}_1 &\leq
\frac{1}{2}\xi^T[\widehat{C}\hat{\rho} G \otimes (QA+A^TQ)]\xi-\frac{\hat{\lambda}_0}{12}\sum_{i=1}^N
(c_i^2\rho_i^2+2\hat{\alpha}) \xi^T_iQBB^TQ\xi_i\\
&\quad+\frac{\hat{\lambda}_0}{24}\sum_{i=1}^N\varphi_i[-\tilde{c}_i^2+(\alpha-1)^2]
+\xi^T (\widehat{C}\hat{\rho} G \mathcal {L}_1\otimes QB)\widetilde{\omega}.
\end{aligned}
\end{equation}
Note that
\begin{equation}\label{lyas70}
\begin{aligned}
2\xi^T (\widehat{C}\hat{\rho} G \mathcal {L}_1\otimes QB)\omega&=2\xi^T (\sqrt{\frac{\hat{\lambda}_0}{12}}\widehat{C}\hat{\rho}\otimes QB)(\sqrt{\frac{12}{\hat{\lambda}_0}}G \mathcal {L}_1\otimes I)\widetilde{\omega}\\
&\leq \frac{\hat{\lambda}_0}{12}\xi^T (\widehat{C}^2\hat{\rho}^2\otimes QBB^TQ)\xi+\frac{12}{\hat{\lambda}_0 }\|(G \mathcal {L}_1\otimes I)\widetilde{\omega}\|^2\\
&\leq \frac{\hat{\lambda}_0}{12}\xi^T (\widehat{C}^2\hat{\rho}^2\otimes QBB^TQ)\xi+\frac{12}{\hat{{\lambda}}_0}\sigma_{\max}^2(G \mathcal {L}_1)\sum_{i=2}^N(\upsilon_i+\upsilon_1)^2,
\end{aligned}
\end{equation}
where we have used \dref{omega} to get the last inequality.
Then, substituting \dref{lyas70} into \dref{lyas7} gives
\begin{equation}\label{lyas701}
\begin{aligned}
\dot{V}_1 &\leq
\frac{1}{2}\xi^T[\widehat{C}\hat{\rho} G \otimes (QA+A^TQ)\\
&\quad-\frac{\hat{\lambda}_0}{24}(\widehat{C}^2\hat{\rho}^2+4\hat{\alpha}I)\otimes QBB^TQ)] \xi
-\frac{\hat{\lambda}_0}{24}\sum_{i=2}^N\varphi_i\tilde{c}_i^2+\Pi\\
&\leq
\frac{1}{2} W(\xi)-\frac{\hat{\lambda}_0}{24}\sum_{i=1}^N\varphi_i\tilde{c}_i^2+\Pi,
\end{aligned}
\end{equation}
where we have used the assertion that
$\frac{\hat{\lambda}_0}{12}(\widehat{C}^2\hat{\rho}^2+4\hat{\alpha}I)\geq\frac{\hat{\lambda}_0}{6}\sqrt{\hat{\alpha}}\widehat{C}\hat{\rho}\geq \widehat{C}\hat{\rho} G$
if $\sqrt{\hat{\alpha}}I\geq\frac{6}{\hat{\lambda}_0}G$ to get the last inequality, $\Pi$ is defined as in \dref{kk}, and
$$\begin{aligned}
W(\xi)&\triangleq\xi^T[\widehat{C}\hat{\rho} G \otimes (QA+A^TQ-QBB^TQ)]\xi\\
&=-\xi^T(\widehat{C}\hat{\rho} G \otimes I)\xi\\
&\leq0.
\end{aligned}$$
Therefore, we can verify that
$\frac{1}{2}W(\xi)-\frac{\hat{\lambda}_0}{24}\sum_{i=2}^N\varphi_i\tilde{c}_i^2$ is negative definite.
In virtue of the results in \cite{corless1981continuous},
we get that both the consensus error $\xi$
and the adaptive gains $c_i$
are uniformly ultimately bounded.

Note that \dref{lyas7} can be rewritten into
\begin{equation}\label{lyas8}
\begin{aligned}
\dot{V}_1 &\leq -\delta V_1+\delta V_1+\frac{1}{2}W(\xi)-\frac{\hat{\lambda}_0}{24}\sum_{i=2}^N\varphi_i\tilde{c}_i^2+\Pi.
\end{aligned}
\end{equation}
Because $\rho_i$ are monotonically
increasing and satisfy $\rho_i(s)\geq1$ for $s>0$, as shown in \dref{lyas4}, we have
\begin{equation}\label{lyas9}
\begin{aligned}
\sum_{i=2}^N c_i q_i \int_0^{\xi_i^TQ\xi_i}\rho_i(s)ds
&\leq \sum_{i=2}^Nc_i q_i\rho_i\xi_i^TQ\xi_i\\
&=\xi^T(\widehat{C}\hat{\rho} G \otimes Q)\xi.
\end{aligned}
\end{equation}
Substituting \dref{lyas9} into \dref{lyas8} yields
\begin{equation}\label{lyas10}
\begin{aligned}
\dot{V}_1 &\leq -\delta V_1+\frac{1}{2}\tilde{W}(\xi)-\frac{\hat{\lambda}_0}{24}\sum_{i=2}^N(\varphi_i-\delta)\tilde{c}_i^2\\
&\quad-\frac{\tau-\delta}{2}\xi^T(\widehat{C}\hat{\rho} G\otimes Q)\xi+\Pi,
\end{aligned}
\end{equation}
where
$\tilde{W}(\xi)\triangleq\xi^T[\widehat{C}\hat{\rho} G \otimes(-I+\tau Q)]\xi.
$
Because $\tau= \frac{1}{\lambda_{\max}(Q)}$,
we can obtain that
$\tilde{W}(\xi)\leq0$. Then, it follows from \dref{lyas10} that
\begin{equation}\label{lyas101}
\begin{aligned}
\dot{V}_1 &\leq -\delta V_1-\frac{\tau-\delta}{2}\lambda_{\min}(Q)(\underset{i=2,\cdots,N}{\min} q_i)\|\xi\|^2+\Pi,
\end{aligned}
\end{equation}
where we have used the facts that $\varphi_i\geq\delta$, $\delta<\tau$, $i=2,\cdots,N$,
$\widehat{C}\geq I$, $\hat{\rho} \geq I$,
and $G>0$.
Obviously, it follows from \dref{lyas10} that $\dot{V}_1\leq -\delta V_1$ if
$\|\xi\|^2>\frac{2\Pi}{(\tau-\delta)\lambda_{\min}(Q)\underset{i=2,\cdots,N}{\min}  q_i }.$
Then,
we can get that if $\delta<\tau$ then
$\xi$ exponentially converges to the residual set $\mathcal {D}_1$ in \dref{d2k}
with a convergence rate faster than ${\rm{e}}^{-\delta t}$.
\end{proof}

\begin{remark}
It is well known that there exists a unique solution $Q>0$ to
the ARE \dref{alg1} if $(A,B)$ is stabilizable \cite{zhou1998essentials}. Therefore, a
sufficient condition for the existence of an adaptive protocol \dref{ssda}
satisfying Theorem \ref{s2thm1} is that $(A,B)$ is stabilizable.
The consensus protocol \dref{ssda} can also be equivalently designed
by solving the linear matrix inequality: $AP+PA^T-2BB^T<0$,
as in \cite{li2010consensus,li2014TAC}.
In this case, the parameters in \dref{ssda} can be chosen
as $K=-B^TP^{-1}$, $\Gamma=P^{-1}BB^TP^{-1}$,
and $\rho_i=(1+\xi_i^TP^{-1}\xi_i)^3$.
Similar to the adaptive protocol \dref{ssda0} in Lemma \ref{lems2},
the adaptive protocol \dref{ssda} in Theorem \ref{s2thm1}, depending only
on the agent dynamics and the relative states of neighboring
agents, can be constructed and implemented in a fully distributed fashion.
\end{remark}

\begin{remark}
Theorem \ref{s2thm1} shows that
the modified adaptive protocol \dref{ssda}
can ensure the ultimate boundedness of the consensus error $\xi$
and the adaptive gains $c_{i}$ for the agents in
\dref{1c}, implying that \dref{ssda} is indeed robust in the presence
of bounded external disturbances.
From \dref{d2k}, it can be observed
that the upper bound of the consensus error $\xi$
depends on the communication graph, the upper bounds of the external disturbances,
and the parameters $\varphi_i$ of the adaptive protocol \dref{ssda}.
Roughly speaking, $\varphi_i$ should be chosen to be relatively small
in order to ensure a smaller bound for
$\xi$.
\end{remark}

\section{Distributed Robust Adaptive Protocols for Strongly Connected Graphs in The Presence of External Disturbances}

The results in the previous section are applicable to the case where there exists a leader.
In this section, we extend to consider the case where the communication graph among
the agents is directed and does not contain a leader.

The dynamics of the $N$ agents are still described by \dref{1c}.
The communication graph $\mathcal {G}$ among the $N$ agents is
assumed to be strongly connected in this section.

Based on the relative states of neighboring agents,
we propose the following distributed adaptive consensus protocol:
\begin{equation}\label{ssdanl2}
\begin{aligned}
u_i &=d_i\rho_i(\zeta_i^TQ\zeta_i)K\zeta_i,\\
\dot{d}_i &=-\varphi_i(d_i-1)+\zeta_i^T\Gamma\zeta_i,\quad i=1,\cdots,N,
\end{aligned}
\end{equation}
where $\zeta_i\triangleq\sum_{j=1}^Na_{ij}(x_i-x_j)$,
$d_i(t)$ denotes the
time-varying coupling gain associated with the $i$-th agent
with $d_i(0)\geq1$,
$\varphi_i$, $i=1,\cdots,N$, are small positive constants,
and the rest of the variables are defined as in \dref{ssda}.

Let $\zeta=[\zeta_1^T,\cdots,\zeta_N^T]^T$ and $x=[x_1^T,\cdots,x_N^T]^T$. Then, $\zeta=(\mathcal {L}\otimes I_n)x$,
where $\mathcal {L}$ denotes the Laplacian matrix associated with $\mathcal {G}$.
Since $\mathcal {G}$ is strongly connected,
it is well known via Lemma \ref{lemma1} that the consensus problem
is solved if and only if $\zeta$ asymptotically
converges to zero.
Hereafter, we refer to $\zeta$ as the consensus error.
In virtue of \dref{1c} and \dref{ssdanl2}, it is not difficult to get that
$\zeta$ and $d_i$ satisfy the following dynamics:
\begin{equation}\label{netssnl2}
\begin{aligned}
\dot{\zeta}
&= [I_N\otimes A+\mathcal {L}\widehat{D}\tilde{\rho}(\zeta)\otimes BK]\zeta+(\mathcal {L} \otimes B)\omega,\\
\dot{d}_i &=-\varphi_i(d_i-1)+\zeta_i^T\Gamma\zeta_i,
\end{aligned}
\end{equation}
where $\tilde{\rho}(\zeta)\triangleq{\rm{diag}}(\rho_1(\zeta_1^TQ \zeta_1),\cdots,\rho_N(\zeta_N^TQ \zeta_N))$,
$\omega\triangleq [\omega_1^T,\cdots,\omega_N^T]^T$, and
$\widehat{D}\triangleq{\rm{diag}}(d_1,\cdots,d_N)$.

\begin{theorem}\label{thm2}
Suppose that the communication graph $\mathcal
{G}$ is strongly connected and Assumption 1 holds.
Then, both the consensus error $\zeta$
and the coupling gains $d_i$, $i=1,\cdots,N$, in \dref{netssnl2},
under the adaptive protocol \dref{ssdanl2} with $K$, $\Gamma$, and $\rho_i$ designed as in Theorem \ref{s2thm1},
are uniformly ultimately bounded.
Furthermore, if $\psi_i$ is chosen to be small enough such that
$\varepsilon\triangleq\underset{i=1,\cdots,N}{\min}\psi_i<\tau\triangleq \frac{1}{\lambda_{\max}(Q)}$,
then $\zeta$ exponentially converges to the residual set
\begin{equation}\label{d1k}
\mathcal {D}_2\triangleq \left\{\zeta:\|\zeta\|^2\leq
\frac{2\Xi}{(\tau-\varepsilon)\lambda_{\min}(Q)\underset{i=1,\cdots,N}{\min} r_i}\right\},
\end{equation}
where $[r_1,\cdots,r_{N}]^T$ is the positive left eigenvector of $\mathcal
{L}$ associated with the zero eigenvalue, $\lambda_2(\widehat{\mathcal {L}})$
denotes the smallest nonzero eigenvalue of $\widehat{\mathcal {L}}\triangleq R\mathcal {L}+\mathcal {L}^TR$,  $R\triangleq{\rm{diag}}(r_1,\cdots,r_{N})>0$,
$\beta=\frac{72N^2}{\lambda_2(\widehat{\mathcal {L}})}\underset{i=1,\cdots,N}{\max} r_i^2+\underset{i=1,\cdots,N}{\max}\frac{2r_i^3N^3}{\lambda_2(\widehat{\mathcal {L}})^3}$,
and
\begin{equation}\label{kk}
\Xi\triangleq\frac{\lambda_2(\widehat{\mathcal {L}})}{24N}\sum_{i=1}^N\varphi_i(\alpha-1)^2+\frac{12N}{\lambda_2(\widehat{\mathcal {L}})}\sigma_{\max}^2(R \mathcal {L})\sum_{i=1}^N\upsilon_i^2.
\end{equation}
\end{theorem}

\begin{proof}
Consider the following Lyapunov function candidate:
\begin{equation}\label{lyaz1}
V_2=\sum_{i=1}^N\frac{d_ir_i}{2}\int_0^{\zeta_i^TP \zeta_i}\rho_i(s)ds
+\frac{\lambda_2(\widehat{\mathcal {L}})}{24N}\sum_{i=1}^N\tilde{d}_i^2,
\end{equation}
where $\tilde{d}_i=d_i-\beta$.
Similarly as shown in \dref{cc0}, it is easy to see
that $d_i(t)\geq1$ for $t>0$. Furthermore, by noting that
$\rho_i(\cdot)$ are monotonically increasing functions
satisfying $\rho_i(s)\geq1$ for $s>0$, it is not difficult to see
that $V_2$ is positive definite.

The time derivative of $V_2$
along the trajectory of \dref{netssnl2} is given by
\begin{equation}\label{lyaz2}
\begin{aligned}
\dot{V}_2 &=
\sum_{i=1}^Nd_ir_i \rho_i(\zeta_i^TP\zeta_i)\zeta_i^TQ\dot{\zeta}_i\\
&\quad+
\sum_{i=1}^N\frac{\dot{d}_ir_i}{2 }\int_0^{\zeta_i^TQ\zeta_i}\rho_i(s)ds\\
&\quad
+\frac{\lambda_2(\widehat{\mathcal {L}})}{12N}\sum_{i=1}^N(d_i-\beta)[-\varphi_i(d_i-1)+\zeta_i^T\Gamma\zeta_i].
\end{aligned}
\end{equation}

By using \dref{netssnl2} and making some mathematical manipulations,
we can get that
\begin{equation}\label{lyasnl3}
\begin{aligned}
\sum_{i=1}^N d_ir_i\rho_i\zeta_i^TQ^{-1}\dot{\zeta}_i&=\zeta^T(\widehat{D}\tilde{\rho} R \otimes Q )\dot{\zeta}\\
&=\frac{1}{2}\zeta^T[\widehat{D}\tilde{\rho} R \otimes (QA+A^TQ)\\&\quad
-\widehat{D}\tilde{\rho} \widehat{\mathcal {L}}\widehat{D}\tilde{\rho} \otimes Q BB^TQ ]\zeta
+\zeta^T (\widehat{D}\tilde{\rho} R \mathcal {L} \otimes PB)\omega.
\end{aligned}
\end{equation}
Let $\bar{\zeta}=(\widehat{D}\tilde{\rho}\otimes I_n)\zeta$.
By the definitions of $\zeta$ and $\bar{\zeta}$, we have
$$\begin{aligned}\bar{\zeta}^T(\widehat{D}^{-1}\tilde{\rho}^{-1}r\otimes I_n)&=\zeta^T(r \otimes I_n)\\
&=x^T(\mathcal {L}^Tr \otimes I_n)=0,
\end{aligned}$$
where we have used fact that $r^T\mathcal {L}=0$.
Since every entry of $r$ is positive,
it is easy to see that every entry of $\widehat{D}^{-1}\tilde{\rho}^{-1}r\otimes I_n$ is also positive. In light of Lemma \ref{lemma2},
we get that
\begin{equation}\label{lyasnl3x}
\begin{aligned}
\bar{\zeta}^T(\widehat{\mathcal {L}}\otimes I_n)\bar{\zeta}&>\frac{\lambda_2(\widehat{\mathcal {L}})}{N}\bar{\zeta}^T\bar{\zeta}\\
&=\frac{\lambda_2(\widehat{\mathcal {L}})}{N}\zeta^T(\widehat{D}^2\tilde{\rho}^2\otimes I_n)\zeta.
\end{aligned}
\end{equation}
Substituting \dref{lyasnl3x} into \dref{lyasnl3} gives
\begin{equation}\label{lyasnl3y}
\begin{aligned}
\sum_{i=1}^N d_ir_i\rho_i\zeta_i^TQ\dot{\zeta}_i
&\leq\frac{1}{2}\zeta^T[\widehat{D}\tilde{\rho} R \otimes (QA+A^TQ)-\frac{\lambda_2(\widehat{\mathcal {L}})}{N}
\widehat{D}^2\tilde{\rho}^2 \otimes QBB^TQ]\zeta \\&\quad+\zeta^T (\widehat{D}\tilde{\rho} R \mathcal {L} \otimes QB)\omega.
\end{aligned}
\end{equation}

Similar to \dref{lyas4}, we can obtain that
\begin{equation}\label{lyaz4}
\begin{aligned}
\sum_{i=1}^N \dot{d}_ir_i \int_0^{\zeta_i^TQ\zeta_i}\rho_i(s)ds\leq \sum_{i=1}^N[\frac{r_i^3N^2}{3\lambda_2(\widehat{\mathcal {L}})^2}+\frac{2\lambda_2(\widehat{\mathcal {L}})}{3N}\rho_i^2]\zeta_i^TQBB^TQ\zeta_i,
\end{aligned}
\end{equation}

Substituting \dref{lyasnl3y} and \dref{lyaz4} into \dref{lyaz2} yields
\begin{equation}\label{lyaz5}
\begin{aligned}
\dot{V}_2 &\leq
\frac{1}{2}\zeta^T[\widehat{D}\tilde{\rho} G \otimes (QA+A^TQ)]\zeta\\
&\quad-\sum_{i=1}^N
[\frac{\lambda_2(\widehat{\mathcal {L}})}{N}(\frac{1}{2}d_i^2\rho_i^2-\frac{1}{12}d_i-\frac{1}{3}\rho_i^2)
\\
&\quad+\frac{1}{12N}(\beta\lambda_2(\widehat{\mathcal {L}})-
\frac{2r_i^3N^3}{\lambda_2(\widehat{\mathcal {L}})^2})]\zeta^T_iQBB^TQ\zeta_i\\
&\quad +\frac{\lambda_2(\widehat{\mathcal {L}})}{24N}\sum_{i=1}^N\varphi_i[-\tilde{d}_i^2+(\beta-1)^2]
+\zeta^T (\widehat{D}\tilde{\rho} R \mathcal {L}\otimes QB)\omega,
\end{aligned}
\end{equation}
where we have used \dref{sskk2}.
Similar to \dref{lyas70}, it is easy to verify that
\begin{equation}\label{lyaz6}
\begin{aligned}
2\zeta^T (\widehat{D}\tilde{\rho} R \mathcal {L} \otimes QB)\omega
\leq \frac{\lambda_2(\widehat{\mathcal {L}})}{12N}\zeta^T (\widehat{D}^2\tilde{\rho}^2\otimes QBB^TQ)\zeta+\frac{12N}{\lambda_2(\widehat{\mathcal {L}})}\sigma_{\max}^2(R \mathcal {L})\sum_{i=1}^N\upsilon_i^2.
\end{aligned}
\end{equation}
Choose $\beta=2\tilde{\beta}+\underset{i=1,\cdots,N}{\max}\frac{2r_i^3N^3}{\lambda_2(\widehat{\mathcal {L}})^3}$, where $\tilde{\beta}=\frac{ 36N^2}{\lambda_2(\widehat{\mathcal {L}})}\underset{i=1,\cdots,N}{\max} r_i^2$. Substituting \dref{lyaz6} into \dref{lyaz5} gives
\begin{equation}\label{lyaz7}
\begin{aligned}
\dot{V}_2 &\leq
\frac{1}{2}\zeta^T[\widehat{D}\tilde{\rho} R \otimes (QA+A^TQ)]\zeta\\
&\quad-\frac{\lambda_2(\widehat{\mathcal {L}})}{24N}
(\widehat{D}^2\rho^2+4\tilde{\beta}I)\otimes \zeta^T_iQBB^TQ\zeta_i\\
&\quad -\frac{\lambda_2(\widehat{\mathcal {L}})}{24N}\sum_{i=1}^N\varphi_i \tilde{d}_i^2+\Xi \\
&\leq
\frac{1}{2} Z(\zeta)-\frac{\lambda_2(\widehat{\mathcal {L}})}{24N}\sum_{i=1}^N\varphi_i\tilde{d}_i^2+\Xi,
\end{aligned}
\end{equation}
where $\Xi$ is defined as in \dref{kk},
$$\begin{aligned}
Z(\zeta)&\triangleq\zeta^T[\widehat{D}\tilde{\rho} R \otimes (QA+A^TQ-QBB^TQ)]\zeta\\
&=-\zeta^T(\widehat{D}\tilde{\rho} R \otimes I)\zeta \leq0,
\end{aligned}$$
and to get the last inequality,
we have used the assertion that if $\sqrt{\tilde{\beta}}I\geq\frac{6N}{\lambda_2(\widehat{\mathcal {L}})}R$, then
$\frac{\lambda_2(\widehat{\mathcal {L}})}{12N}(\widehat{D}^2\tilde{\rho}^2+4\tilde{\beta}I)\geq\frac{\lambda_2(\widehat{\mathcal {L}})}{6N}\sqrt{\tilde{\beta}}\widehat{D}\tilde{\rho}\geq \widehat{D}\tilde{\rho}R.$
Therefore, we can verify that
$\frac{1}{2} Z(\zeta)-\frac{\lambda_2(\widehat{\mathcal {L}})}{24N}\sum_{i=1}^N\varphi_i\tilde{d}_i^2$ is negative definite.
In virtue of Lemma \ref{lemma5},
we get that both the consensus error $\zeta$
and the adaptive gains $d_i$
are uniformly ultimately bounded.

Note that \dref{lyaz7} can be rewritten into
\begin{equation}\label{lyaz8}
\begin{aligned}
\dot{V}_2 &\leq -\varepsilon V_2+\varepsilon V_2+\frac{1}{2} Z(\zeta)-\frac{\lambda_2(\widehat{\mathcal {L}})}{24N}\sum_{i=1}^N\varphi_i\tilde{d}_i^2+\Xi.
\end{aligned}
\end{equation}
As shown in \dref{lyas9}, we have
\begin{equation}\label{lyaz9}
\begin{aligned}
\sum_{i=1}^N d_ir_i\int_0^{\zeta_i^TQ\zeta_i}\rho_i(s)ds
\leq \zeta^T(\widehat{D}\tilde{\rho} R \otimes Q)\zeta.
\end{aligned}
\end{equation}
Substituting \dref{lyaz9} into \dref{lyaz8} yields
\begin{equation}\label{lyaz10}
\begin{aligned}
\dot{V}_2 &\leq -\varepsilon V_2+\frac{1}{2}\tilde{Z}(\zeta)-\frac{\lambda_2(\widehat{\mathcal {L}})}{24N}\sum_{i=1}^N(\varphi_i-\varepsilon)\tilde{d}_i^2\\
&\quad-\frac{\tau-\varepsilon}{2}\zeta^T(\widehat{D}\tilde{\rho} R\otimes Q)\zeta+\Xi,
\end{aligned}
\end{equation}
where
$\tilde{Z}(\zeta)\triangleq\zeta^T[\widehat{D}\tilde{\rho} R \otimes(-I+\tau Q)]\zeta.
$
Because $\tau= \frac{1}{\lambda_{\max}(Q)}$,
we can obtain that
$\tilde{Z}(\zeta)\leq0$. Then, it follows from \dref{lyaz10} that
\begin{equation}\label{lyas10}
\begin{aligned}
\dot{V}_2 &\leq -\varepsilon V_2-\frac{\tau-\varepsilon}{2}\lambda_{\min}(Q)(\underset{i=1,\cdots,N}{\min} r_i)\|\zeta\|^2+\Xi,
\end{aligned}
\end{equation}
where we have used the facts that $\varphi_i\geq\varepsilon$, $\varepsilon<\tau$,
$\widehat{D}\geq I$, $\tilde{\rho} \geq I$,
and $R>0$.
Obviously, it follows from \dref{lyas10} that $\dot{V}_2\leq -\varepsilon V_2$ if
$\|\zeta\|^2>\frac{2\Xi}{(\tau-\varepsilon)\lambda_{\min}(Q)\underset{i=1,\cdots,N}{\min} r_i}.$
Then,
we can get that if $\varepsilon\leq\tau$ then
$\zeta$ exponentially converges to the residual set $\mathcal {D}_2$ in \dref{d2k}
with a convergence rate faster than ${\rm{e}}^{-\varepsilon t}$.
\end{proof}

In the robust adaptive protocol \dref{ssdanl2}, the term $-\varphi_i(d_i-1)$ is inspired by the $\sigma$ modification technique,
which is vital to ensuring the ultimate boundedness of the consensus error $\zeta$ and the adaptive gains $d_{i}$ in the presence of external
disturbances. For the case where the external disturbances in \dref{1c} do not exist, the adaptive protocol \dref{ssdanl2} with the term $-\varphi_i(d_i-1)$ removed, i.e., the following adaptive protocol
\begin{equation}\label{ssdanl20}
\begin{aligned}
u_i &=d_i\rho_i(\zeta_i^TQ\zeta_i)K\zeta_i,\\
\dot{d}_i &=\zeta_i^T\Gamma\zeta_i,\quad i=1,\cdots,N,
\end{aligned}
\end{equation}
can ensure the asymptotical convergence of the consensus error $\zeta$. This is summarized in the following corollary.

\begin{corollary}\label{cor1}
For the $N$ agents described by $\dot{x}_i =Ax_i+Bu_i$,
$i=1,\cdots,N$, whose communication graph $\mathcal
{G}$ is strongly connected,
the consensus error $\zeta$
under the adaptive protocol \dref{ssdanl2} with $K$, $\Gamma$, and $\rho_i$ designed as in Theorem \ref{s2thm1}
asymptotically converges to zero.
Moreover, each coupling gain
$d_{i}$ converges to some finite steady-state value.
\end{corollary}

The above corollary can be proved by following similar steps in the proof of Theorem \ref{thm2}.
The adaptive protocol \dref{ssdanl20} complements the adaptive protocol \dref{ssda0} in \cite{li2014TAC}
which are applicable to directed graphs with a leader.

\begin{remark}
Compared to the previous works \cite{li2014TAC,mei2014consensus} which also present distributed adaptive protocols for directed graphs,
the main contribution of this paper is that distributed robust adaptive protocols are presented,
which can exclude the parameter drift phenomenon encountered
by the adaptive protocols in \cite{li2014TAC,mei2014consensus} in the presence of external disturbances.
Besides, the agents are restricted to be second-order integrators in \cite{mei2014consensus}.
The ultimate boundedness of both the consensus errors and the adaptive gains is shown
and the upper bounds of the consensus errors is given, which are far from being easy.
\end{remark}


\section{Simulation Example}

\begin{figure}[htbp]
\centering
\includegraphics[width=0.5\linewidth]{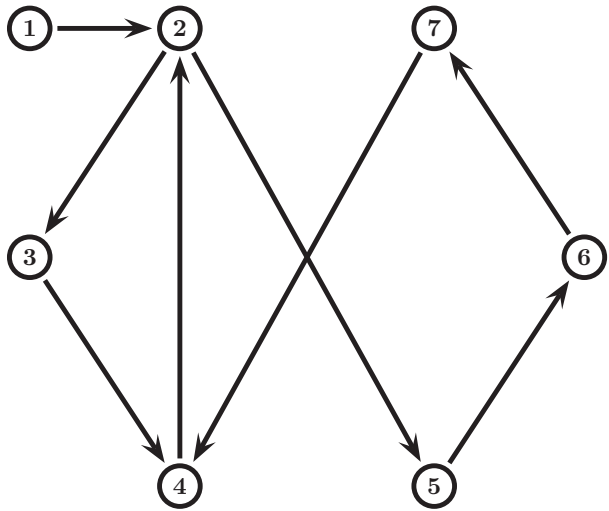}
\caption{The directed communication graph. }
\end{figure}

In this section,
a simulation example is provided for illustration.

Consider a network of double integrators,
described by \dref{1c}, with
$$x_i=\left[\begin{matrix}x_{i1}\\x_{i2}\end{matrix}\right],\quad A=\left[\begin{matrix}0 & 1\\ 0 & 0\end{matrix}\right],\quad
B
=\left[\begin{matrix} 0 \\ 1\end{matrix}\right].$$ 
For illustration, the disturbances associated
with the agents are assumed to be
$\omega_1=0$,
$\omega_2=0.2\sin(t)$
, $\omega_3=0.1\sin(t),$
$\omega_4=0.2\cos(2t)$,
$\omega_5=-0.3{\rm{exp}}^{-2t}$,
$\omega_6=-0.2\sin(x_{51})$,
and $\omega_7=0$.
The communication
graph is given as in Fig. 1, where the vertex indexed by
1 is the leader which is only accessible to the vertex indexed by 2.
It is easy to verify that the graph in Fig. 1 satisfies
Assumption 2.
%

Solving the ARE \dref{alg1} by using MATLAB
gives a solution
$Q =\left[\begin{smallmatrix}1.7321 &  1\\
   -1  &  1.7321\end{smallmatrix}\right].$
Thus, the feedback gain matrices in
\dref{ssda} are obtained as
$$K=-\begin{bmatrix}1  &1.7321\end{bmatrix}, \quad
\Gamma =\begin{bmatrix}1 &   1.7321\\
    1.7321  &  3\end{bmatrix}.$$
To illustrate Theorem \ref{s2thm1}, let $\varphi_i=0.02$ in \dref{ssda} and
the initial states $c_i(0)$ be randomly chosen within the interval $[1,3]$.
The consensus errors $\xi_i$, $i=2,\cdots,7$, of the double integrators, defined as in \dref{conerr},
and the coupling weights
$c_{i}$ associated with the followers, under the adaptive protocol \dref{ssda} with
$K$, $\Gamma$, and $\rho_i$ chosen as in Theorem \ref{s2thm1},
are depicted in in Figs. 2 and 3, respectively,
both of which are clearly bounded.

\begin{figure}[htbp]
\centering
\includegraphics[width=0.7\linewidth,height=0.4\linewidth]{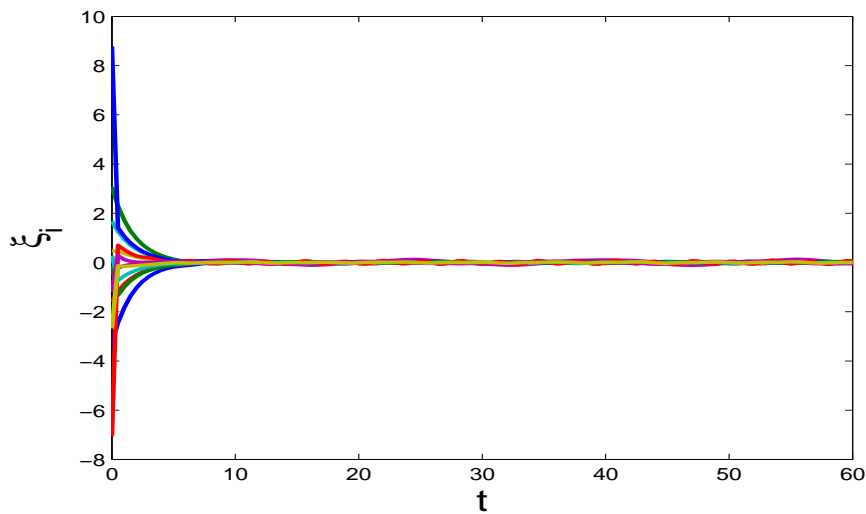}
\caption{The consensus errors $\xi_i$, $i=2,\cdots,7$, of double integrators under the protocol \dref{ssda}. }
\end{figure}

\begin{figure}[htbp]
\centering
\includegraphics[width=0.7\linewidth,height=0.4\linewidth]{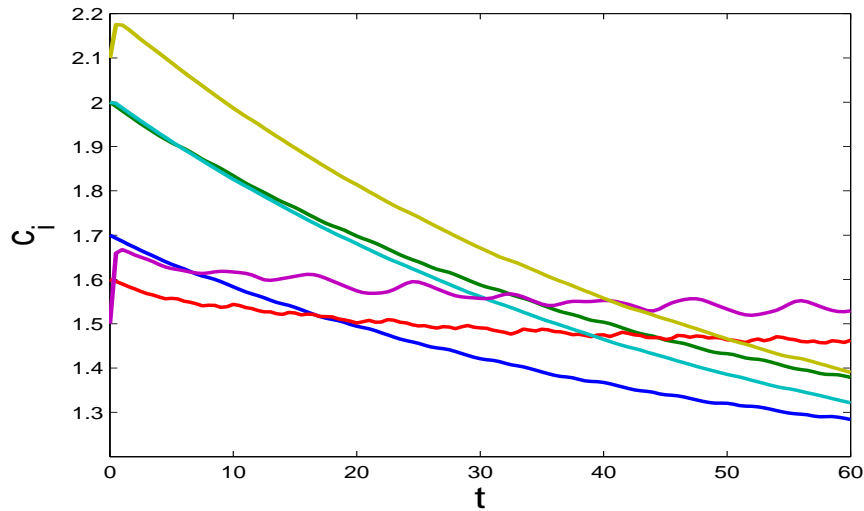}
\caption{The coupling gains $c_{i}$ in \dref{ssda}. }
\end{figure}

\section{Conclusion}
In this paper, we have presented distributed adaptive consensus protocols to achieve consensus
for linear multi-agent systems with directed graphs which are strongly connected or contain a directed spanning
tree with a leader as the root node. Specifically, distributed robust adaptive protocols are designed,
which can guarantee the ultimate boundedness of both the consensus error
and the adaptive gains in the presence of external disturbances. Note that the design of these adaptive protocols
depends only on the agent dynamics and the relative state information of neighboring agents,
which thereby can be done by each agent
in a truly distributed fashion. Interesting future works include designing distributed adaptive consensus
protocols using only relative output information.


\end{document}